\documentclass[pdftex, 11pt]{article}
\usepackage[hiresbb]{graphicx}
\usepackage{graphicx}
\usepackage{subcaption}
\usepackage{wrapfig}
\usepackage{float}
\usepackage{amsmath}
\usepackage{amssymb}
\usepackage{amsthm}
\usepackage{authblk}
\usepackage{authblk}
\usepackage[top=30truemm,bottom=30truemm,left=20truemm,right=20truemm]{geometry}
\usepackage[titletoc,toc,title]{appendix}
\usepackage{latexsym}
\usepackage{algorithm}
\usepackage{algorithmic}
\usepackage{color}
\usepackage{natbib}
\bibpunct[:]{(}{)}{,}{a}{}{,}
\usepackage{mathtools}
\usepackage[english]{babel}
\newtheorem{thm}{Theorem}[section]

\providecommand{\keywords}[1]{\textbf{\textit{Keywords:}} #1}
\newcommand{\ty}{t}
\newcommand{\tx}{s}
\newcommand{\mdim}{d}
\newcommand{\vx}{u}
\newcommand{\vy}{v}
\newcommand{\parav}{\xi}
\newcommand{\hpi}{h_\pi}
\newcommand{\dd}{\mathrm{d}}
\title{Predictive densities for multivariate normal models \\based on extended models and shrinkage Bayes methods}
\date{}
\author[1]{Michiko Okudo}
\author[1,2]{Fumiyasu Komaki}
\affil[1]{\small Department of Mathematical Informatics

Graduate School of Information Science and Technology

The University of Tokyo

7-3-1 Hongo, Bunkyo-ku, Tokyo 113-8656, JAPAN}

\affil[2]{RIKEN Center for Brain Science

2-1 Hirosawa, Wako City, Saitama 351-0198, JAPAN}

\begin{document}
\maketitle

\begin{center}
Abstract    
\end{center}

We investigate predictive densities for multivariate normal models with unknown mean vectors
and known covariance matrices.
Bayesian predictive densities based on shrinkage priors often have complex representations,
although they are effective in various problems.
We consider extended normal models
with mean vectors and covariance matrices as parameters,
and adopt predictive densities that belong to the extended models including the original normal model.
We adopt predictive densities that are optimal with respect to the posterior Bayes risk in the extended models.
The proposed predictive density based on a superharmonic shrinkage prior is shown to dominate the Bayesian predictive density
based on the uniform prior under a loss function based on the Kullback--Leibler divergence.
Our method provides an alternative to the empirical Bayes method, which is  widely used to construct tractable predictive densities.

\vspace{0.5cm}


\section{Introduction}
Suppose that we have independent observations $x_1, \dots, x_n$ from a $\mdim$-dimensional multivariate normal model
$\mathrm{N}_\mdim(\mu, I_\mdim)$,~$\mu\in\mathbb{R}^\mdim$.
By sufficiency reduction, it is sufficient to consider the setting in which we have a single observation $x$
distributed according to $\mathrm{N}_\mdim(\mu, \vx I_\mdim)$, 
where $\vx > 0$ is known and fixed.
We address the problem of predicting a future outcome $y$
following a $\mdim$-dimensional multivariate normal distribution
$\mathrm{N}_\mdim(\mu, \vy I_\mdim),~~\mu\in\mathbb{R}^\mdim$, $\vy > 0$
with the same mean vector $\mu$ by using a predictive density $\hat{p}(y \mid x)$ that depends on $x$.
The variance $\vy$ is known and possibly differs from $\vx$.
The performance of a predictive density $\hat{p}(y\mid x)$ is evaluated by the Kullback--Leibler divergence
\[
D\{p(y;\mu, \vy I_\mdim);\hat{p}(y\mid x) \} = \int p(y;\mu, \vy I_\mdim)
\log \frac{p(y;\mu, \vy I_\mdim)}{\hat{p}(y\mid x)}dy,
\]
where $p(y;\mu, \Sigma)~(\mu\in\mathbb{R}^\mdim, \Sigma\in\mathbb{R}^{\mdim \times \mdim})$ is the density of $\mathrm{N}_\mdim(\mu, \Sigma)$.

There are two widely used methods to construct predictive densities: Bayesian predictive densities and plug-in densities.
Bayesian predictive densities are expressed as
\[
p_\pi(y \mid x) = \int p(y;\mu, \vy I_\mdim) p_\pi(\mu \mid x) d\mu,
\]
where $p_\pi(\mu \mid x)$ is the posterior density
\[
p_\pi(\mu \mid x) = \frac{p(x;\mu, \vy I_\mdim)\pi(\mu)}{\int p(x;\mu, \vy I_\mdim)\pi(\mu) d\mu}
\]
based on a prior density $\pi(\mu)$.
Bayesian predictive densities do not belong to a tractable finite-dimensional family
unless a conjugate prior is adopted.
On the other hand, plug-in predictive densities can easily be obtained by plugging
an estimator $\hat{\mu}$ such as maximum likelihood estimators or Bayes estimators, in the unknown parameter $\mu$
of the density $p(y;\mu,vI_d)$ of $y$.
However, Bayesian predictive densities are preferable to plug-in densities in many examples.

Shrinkage methods are effective both in estimation and in prediction for normal models with unknown mean vectors $\mu$.
Bayes estimators based on Stein's prior $\pi_\mathrm{S}(\mu) \propto \|\mu \|^{-(d-2)}$ dominates the maximum likelihood estimator $\hat{\mu}_\mathrm{mle}=x$ when $\mdim \geq 3$ \citep{stein1974}.
Priors that ``shrink'' posterior density to a certain point such as the origin or to a subspace, are called shrinkage priors.
If a function $\pi(\mu)$ satisfies the inequality
\[
\Delta \pi(\mu) := \sum_{i=1}^\mdim \frac{\partial^2}{\partial \mu_i ^2} \pi(\mu) \leq 0,
\]
then $\pi(\mu)$ is said to be superharmonic.
Bayes estimators based on nonconstant superharmonic priors dominate the maximum likelihood estimator \citep{stein1974}.
The density $\pi_\mathrm{S}$ shrinks the posterior to the origin and
satisfies
\[
\Delta \pi_\mathrm{S}(\mu) = - \delta(\mu),
\]
where $\delta$ denotes the Dirac delta function,
in the framework of Schwartz's distribution theory, see e.g.\ \citet{john1978} p.~74.
In this sense, $\pi_\mathrm{S}$ is a superharmonic function.
The maximum likelihood estimator coincides with the Bayes estimator based on the uniform prior $\pi_\mathrm{U}(\mu)\propto 1.$

A parallel result regarding Bayesian prediction is obtained by \citet{komaki2001},
and the Bayesian predictive density based on Stein's prior dominates the Bayesian predictive density based on the uniform prior.
Bayesian predictive densities based on superharmonic priors dominate the Bayesian predictive density based on $\pi_\mathrm{U}$
\citep{george2006}.
Other important shrinkage priors for multivariate normal models with unknown mean include shrinkage priors for regression problems \citep{george2008, kobayashi2008} and singular value shrinkage priors for matrix-variate normal models \citep{matsuda2015}.
The Bayesian predictive density based on $\pi_\mathrm{U}$ has the simple form $\mathrm{N}_d(x, (\vx +\vy )I_d)$.
On the other hand, Bayesian predictive densities based on shrinkage priors generally do not have such simple forms.

The empirical Bayes method is another method of constructing predictive densities with reasonable risk performance and small computational cost.
An empirical Bayes method for approximating a Bayesian predictive density based on Stein's prior is studied by \citet{xu2011}.
Stein's prior is represented as a mixture of normal distributions:
\begin{align}
\pi_\mathrm{S}(\mu) &\propto \|\mu \|^{-(d-2)}
= \frac{2}{\Gamma(d/2 - 1)}\int^\infty_0 (2 \tau)^{-d/2}\exp\left( -\frac{\|\mu\|^2}{2 \tau} \right)d \tau,
\label{integral representation}
\end{align}
where $\Gamma(d/2 - 1)$ denotes the Gamma function.
The representation \eqref{integral representation} is used to construct the Bayesian predictive density based
on Stein's prior in \citet{komaki2001}.
In \citet{xu2011}, Bayesian predictive densities based on a prior
\begin{align*}
\pi(\mu ; \hat{\tau}(x))
= (2\pi\hat{\tau}(x))^{-\mdim/2} \exp \left( -\frac{\|\mu\|^2}{2\hat{\tau}(x)} \right)
\end{align*}
with an estimator $\hat{\tau}(x)$ are constructed.
The predictive density based on the empirical Bayes method is expressed as
\begin{align}
\int p(y ; \mu, \vy I_d) \pi(\mu ; \hat{\tau}(x)) d \mu.
\label{eB}
\end{align}
Therefore, the empirical Bayes method is regarded as an approximation of the full Bayes method
in which a prior is adopted for the hyperparameter $\tau$.
The predictive density \eqref{eB} that is obtained by the empirical Bayes method
is also a normal distribution.

The computational difference between full Bayes and empirical Bayes
lies in tha fact that empirical Bayes methods requires only one plug-in distribution to compute the predictive density.
Approximating $p_\pi(y\mid x)$ by empirical Bayes saves computational cost and
it is effective when predicting densities for many future samples and when $d$ is large.

We present an alternative to the empirical Bayes method
to construct tractable predictive densities based on shrinkage priors.
We consider an ``extended'' model including the original model
$\mathrm{N}_d(\mu, \vy I_d)$ with the fixed $\vy$.
Normal models such as $\mathrm{N}_d(\mu, \parav I_d)~(\parav> 0)$
and $\mathrm{N}_d(\mu, \Sigma)~(\Sigma\in\mathbb{R}^{\mdim \times \mdim})$
are adopted as the extended models.
We denote the predictive densities in those extended models as extended plug-in densities.
The resulting predictive densities are optimal with respect to the posterior Bayes risk in the extended models.
Our method is based on a combination of extended plug-in densities for curved exponential families
\citep{okudo2021} and shrinkage priors.
We can construct predictive densities not only in the normal model $\mathrm{N}_d(\mu, \parav I_d)~ (\parav > 0)$
like empirical Bayes method in \citet{xu2011}, but also in the larger normal model $\mathrm{N}_d(\mu, \Sigma)~ (\Sigma\in\mathbb{R}^{\mdim\times \mdim}, \Sigma \succ 0)$.
This approach could apply to various models besides the normal models that can be embedded
in larger exponential families.

We show that the Kullback--Leibler risk difference of an extended plug-in predictive density
based on a prior and the Bayesian predictive density based on the uniform prior
reduces to the Kullback--Leibler risk difference of the corresponding Bayes estimators
in the limit $1/v \to 0$.

Thus, our
predictive density dominates the Bayesian predictive density based on the uniform prior if the performance of the predictive densities is evaluated in the limit
$1/v \rightarrow 0$.
The numerical simulations suggest that the proposed predictive density 
performs better than the Bayesian predictive density based on the uniform prior even if $1/v$ is not close to $0$.

%
%
\section{Bayes extended estimators}
\subsection{Extended models and estimators}
We investigate extended plug-in densities in extended models as predictive densities.
We consider two extended models:
\begin{align*}
\mathcal{E}_1 &= \{\mathrm{N}_d(\mu, \parav I_d)\mid \mu\in\mathbb{R}^d,\ \parav > 0 \},
\end{align*}
and
\begin{align*}
\mathcal{E}_2 &= \{\mathrm{N}_d(\mu, \Sigma)\mid \mu\in\mathbb{R}^d, \Sigma\in\mathbb{R}^{\mdim \times \mdim}, \Sigma \succ 0 \}
\end{align*}
that includes the original model $\mathcal{P} := \mathrm{N}_d(\mu, \vy I_d)$ with the known $\vy$.
In the first extended model $\mathcal{E}_1$, the variance $\parav$ is a parameter
in contrast that variance $\vy$ is fixed in $\mathcal{P}$.

The second extended model $\mathcal{E}_2$ allows all positive semidefinite covariance matrices $\Sigma$.
The inclusion relation is $\mathcal{P} \subseteq \mathcal{E}_1 \subseteq \mathcal{E}_2.$
Other extended models such as $\{\mathrm{N}_d(\mu, D)\mid \mu\in\mathbb{R}^d,
D:\mdim \mbox{-dimensional diagonal matrix}\}$ can be considered in the same manner.

Although the original model $\mathcal{P}$ is a full exponential family,
it can be formulated as a curved exponential family that is  embedded in the extended models
$\mathcal{E}_1$ or $\mathcal{E}_2$.
Thus, we can choose a predictive density that belongs to $\mathcal{E}_1$ or $\mathcal{E}_2$
instead of the original model $\mathcal{P}$.
For a density function
\begin{align*}
p(y;\theta) =b(y)\exp(s(y)^\top \theta -\psi(\theta))
\end{align*}
of an exponential family $\mathcal{E}$, the expectation parameter is
\begin{align*}
\eta(\theta) = \mathrm{E}[s(y)].
\end{align*}
If the original model $\mathcal{P}$ is a curved exponential family
\begin{align*}
p(y;\omega) =b(y)\exp(s(y)^\top \theta(\omega) -\psi(\theta(\omega)))
\end{align*}
that is embedded in $\mathcal{E}$, the Bayes extended estimator $\eta(\omega)$ is
the posterior mean of $\eta(\omega)$ that minimizes the posterior Bayes risk of $p(y;\hat{\eta})$.
Thus, it is reasonable to consider the extended plug-in densities  that belong to $\mathcal{E}$
and not to $\mathcal{P}$ \citep{okudo2021} for prediction.
In other words, extended plug-in densities with the posterior mean of $\eta$ are 
the closest to the Bayesian predictive densities with respect to the posterior Bayes risk.
We denote the posterior means of $\eta$ based on a prior $\pi$ as the Bayes extended estimator and write it as $\hat{\eta}_\pi$.

\begin{table}[htbp]
  \caption{Extended plug-in distributions and extended Bayes estimators
  with respect to a prior density $\pi$.}
  \label{tb:eplugin}
  \centering
  \begin{tabular}{c|c|c|c}
   \hline 
Extended & Expectation & Bayes & Extended \\[-1mm]
model & parameters & extended estimators & plug-in distribution \\ \hline
$\mathcal{E}_1$: $\mathrm{N}_d(\mu, \parav I_d)$ & $(\mu, d \parav + \mu^\top \mu)$ &
\begin{tabular}{l}
$\hat{\mu}_\pi = \mathrm{E}_\pi[\mu \mid x],$ \\
$\hat{\parav}_\pi = \vy + (\mathrm{E}_\pi[\mu^\top \mu \mid x] -\hat{\mu}_\pi^\top \hat{\mu}_\pi)/d$
\end{tabular}
& $\mathrm{N}_d(\hat{\mu}_\pi, \hat{\parav}_\pi I_d)$ \\ \hline
$\mathcal{E}_2$: $\mathrm{N}_d(\mu, \Sigma)$ & $(\mu, \Sigma + \mu \mu^\top)$ &
\begin{tabular}{l}
$\hat{\mu}_\pi = \mathrm{E}_\pi[\mu\mid x]$ \\
$\hat{\Sigma}_\pi = \vy I_d + \mathrm{E}_\pi[\mu \mu^\top \mid x] -\hat{\mu}_\pi \hat{\mu}_\pi^\top$
\end{tabular}
& $\mathrm{N}_d(\hat{\mu}_\pi, \hat{\Sigma}_{\pi})$ \\
\hline
  \end{tabular}
\end{table}

We obtain the expectation parameters of the extended models $\mathcal{E}_1$ and $\mathcal{E}_2$.
Bayes extended estimators are their posterior means.
The results are shown in Table \ref{tb:eplugin}.

A density function in the extended model $\mathcal{E}_1 = \mathrm{N}_d (\mu, \parav I_d)$ is
\begin{align*}
p(y;\mu,\parav I_d) &= (2\pi \parav)^{-d/2} \exp \{ -(y-\mu)^\top (y-\mu)/{2 \parav} \}\\
&= (2\pi \parav)^{-d/2} \exp\{(\mu/\parav)^\top y - (2 \parav)^{-1}y^\top y -\mu^\top \mu/(2 \parav)\}.
\end{align*}
Thus, the expectation parameters are
\[
\eta = (\mathrm{E}[y], \mathrm{E}[y^\top y]) = (\mu, d \parav + \mu^\top \mu).
\]
The expectation parameter for a density in $\mathcal{P} \subset \mathcal{E}_1$ 
is $\eta = (\mu, d v + \mu^\top \mu)$, where $\vy$ is known and fixed.
Thus, the posterior mean of $\eta$ is
\begin{align}
\hat{\eta}_\pi = (\mathrm{E}_\pi[\mu\mid x], ~d\vy + \mathrm{E}_\pi[\mu^\top \mu\mid x]),
\label{etapostmean}
\end{align}
where
$\mathrm{E}_\pi[\cdot \mid x]$ denotes the expectation
with respect to the posterior density of $\mu$ based on a prior $\pi$.
Although the prior and posterior densities are probability densities on $\mathcal{P}$,
an extended plug-in distribution with the posterior mean $\mathrm{E}_\pi[\eta \mid x]$ does not belong to $\mathcal{P}$, which consequently has a favourable effect on the predictive performance.

By plugging \eqref{etapostmean} into $\eta = (\mu, d \parav+ \mu^\top \mu)$,
we obtain the extended plug-in density
$\mathrm{N}_d(\hat{\mu}_\pi, \hat{\parav}_\pi I_d)$,
with respect to $\mathcal{E}_1$,
where 
$\hat{\mu}_\pi = \mathrm{E}_\pi[\mu\mid x]$
and
$
\hat{\parav}_\pi = \vy + (\mathrm{E}_\pi[\mu^\top \mu \mid x] -\hat{\mu}_\pi^\top \hat{\mu}_\pi)/d.
$

Similarly, the expectation parameter of the second extended model $\mathcal{E}_2$ is
\[
\eta = (\mathrm{E}[y], \mathrm{E}[y y^\top]) = (\mu, \Sigma + \mu \mu^\top).
\]
Thus, the extended plug-in density with the posterior mean $\hat{\eta}_\pi$ based on $\mathcal{E}_2$ is
$\mathrm{N}_d(\hat{\mu}_\pi, \hat{\Sigma}_{\pi})$,
where
$\hat{\Sigma}_\pi = \vy I_d + \mathrm{E}_\pi[\mu \mu^\top \mid x] -\hat{\mu}_\pi \hat{\mu}_\pi^\top$.

We obtain the extended plug-in densities with respect to the uniform prior
$\pi_\mathrm{U}(\mu)=1$.
As the posterior density with respect to $\pi_\mathrm{U}$ is
\[
p_\mathrm{U}(\mu\mid x) = p(x; \mu, \vx  I_d) \pi_\mathrm{U}(\mu),
\]
we obtain
\[
\hat{\mu}_\mathrm{U} = \mathrm{E}_{\pi_\mathrm{U}}[\mu\mid x] = x
\]
and
\[
\hat{\parav}_\mathrm{U} = \vy + \mathrm{E}_{\pi_\mathrm{U}}[\mu^\top \mu \mid x]/\mdim - x^\top x/ \mdim
= \vx  + \vy.
\]
Thus, the extended plug-in distribution
$\mathrm{N}_d(\hat{\mu}_\mathrm{U}, \hat{\parav}_\mathrm{U} I_d)
= \mathrm{N}_d(\hat{\mu}_\mathrm{U}, (u+v) I_d)$
based on $\mathcal{E}_1$ is identical to the Bayesian predictive density
$\hat{p}_\mathrm{U}$ based on $\pi_\mathrm{U}$.
As it is optimal with respect to the posterior Bayes risk among all distributions and included in $\mathcal{E}_1$,
the extended plug-in density based on $\mathcal{E}_2$ is also identical to
$\mathrm{N}_d(\hat{\mu}_\mathrm{U}, \hat{\parav}_\mathrm{U} I_d)$.

We investigate extended plug-in densities based on shrinkage priors
including Stein's prior $\pi_\mathrm{S} \propto \|\mu \|^{-(d-2)}$.
Although Bayesian predictive densities based on shrinkage priors do not belong to normal models,
extended plug-in densities with Bayesian extended estimators based on shrinkage priors
belong to tractable extended models. 

\subsection{Posterior mean representations}
We evaluate posterior means that were described in the previous subsection.
Let
\[
m_\pi(x) = \int p(x;{\mu},\vx I_d)\pi({\mu})d{\mu},
\]
which is the marginal density of $x$.
The derivatives of model density functions are given by
\begin{align*}
\nabla p(x;\mu,\vx ) &= \frac{1}{\vx }(\mu - x)p(x;\mu,\vx I_d),\\
\nabla^2 p(x;\mu,\vx) &=  - \frac{1}{\vx}p(x;\mu,\vx I_d) I_d
+ \frac{1}{\vx^2}(\mu - x)(\mu - x)^\top p(x;\mu,\vx I_d), \\
\shortintertext{and}
\Delta p(x;\mu,\vx) &= -\frac{\mdim}{\vx}p(x;\mu,\vx I_d) + \frac{1}{\vx^2}(\mu-x)^\top (\mu-x)p(x;\mu,\vx I_d),
\end{align*}
where, for a function $f:\mathbb{R}^d \to \mathbb{R}$,
\begin{align*}
\nabla f(x) &:= \left( \frac{\partial f}{\partial x_1}(x), \dots, \frac{\partial f}{\partial x_d}(x)\right)^\top,~~~
\Delta f(x) := \sum_{i=1}^d \frac{\partial^2 f}{\partial x_i^2}(x)
\end{align*}
and $\nabla^2 f$ is the Hessian matrix whose $(i,j)$ element is
\begin{align*}
(\nabla^2 f(x))_{ij} &:= \frac{\partial^2 f(x)}{\partial x_i \partial x_j}.
\end{align*}

The posterior mean of $\mu$ is evaluated in \citet{brown1971} as
\begin{align}
\hat{\mu}_\pi &= \int\mu \frac{p(x;\mu,\vx I_d)\pi(\mu)}
{\int p(x;\tilde{\mu},\vx I_d)\pi(\tilde{\mu})d\tilde{\mu}}d\mu 
= x + \int(\mu-x)\frac{p(x;\mu,\vx I_d)\pi(\mu)}
{\int p(x;\tilde{\mu},\vx I_d)\pi(\tilde{\mu})d\tilde{\mu}}d\mu 
= x + \vx\nabla \log m_\pi.
\label{pm1}
\end{align}
The posterior mean of $(\mu-x)(\mu-x)^\top$ is expressed as
\begin{align}
\int& (\mu-x)(\mu-x)^\top\frac{p(x;\mu,\vx I_d)\pi(\mu)}
{\int p(x;\tilde{\mu},\vx I_d)\pi(\tilde{\mu})d\tilde{\mu}}d\mu \nonumber \\
&= \int \{ \vx^2 \nabla^2 p(x;\mu,\vx I_d) + \vx p(x;\mu,\vx I_d) I_d \}
\frac{\pi(\mu)}{\int p(x;\tilde{\mu},\vx I_d)\pi(\tilde{\mu})d\tilde{\mu}}d\mu \nonumber\\
&= \vx^2 \frac{\nabla^2 m_\pi}{m_\pi} + \vx I_d.
\label{pm2}
\end{align}
Thus, the posterior mean of $(\mu-x)^\top(\mu-x)$ is
\begin{align}
&\int (\mu-x)^\top(\mu-x)\frac{p(x;\mu,\vx I_d)\pi(\mu)}
{\int p(x;\tilde{\mu},\vx I_d)\pi(\tilde{\mu})d\tilde{\mu}}d\mu 
= \vx^2 \frac{\Delta m_\pi}{m_\pi} + \mdim \vx.
\label{pm3}
\end{align}

From \eqref{pm1}, \eqref{pm2}, and \eqref{pm3},
the estimators $\hat{\parav}_\pi$ and $\hat{\Sigma}_{\pi}$ are given by
\begin{align}
\hat{\parav}_\pi &=\vy + (\mathrm{E}_\pi[\mu^\top \mu \mid x]
-\hat{\mu}_\pi^\top\hat{\mu}_\pi)/d \nonumber\\
&= \vy + \mathrm{E}_\pi[(\mu-x)^\top (\mu-x) \mid x]/d - (\hat{\mu}_\pi-x)^\top (\hat{\mu}_\pi-x)/d \nonumber\\
&= \vx + \vy + \hpi(x),
\label{hpi}
\end{align}
where
\[
\hpi(x) \coloneqq \frac{\vx^2}{\mdim}\frac{\Delta m_\pi}{m_\pi} - \frac{\vx^2}{\mdim}\frac{\|\nabla m_\pi\|^2}{m_\pi^2},
\]
and
\begin{align}
\hat{\Sigma}_\pi
&= \vy I_d + \mathrm{E}_\pi[\mu \mu^\top \mid x] -\hat{\mu}_\pi^\top \hat{\mu}_\pi \nonumber\\
&= \vy I_d + \mathrm{E}_\pi[(\mu-x) (\mu-x)^\top \mid x] -(\hat{\mu}_\pi-x) (\hat{\mu}_\pi-x)^\top \nonumber\\
&= (\vx + \vy) I_d + H_\pi(x),
\label{Sigmapi}
\end{align}
where
\[
H_\pi(x) \coloneqq \vx^2\left( \frac{\nabla^2 m_\pi}{m_\pi} - \frac{\nabla m_\pi \nabla m_\pi^\top}{m_\pi^2}\right).
\]

Note that $\hat{\parav}_\pi$ is greater than the model variance $\vy$.
If $\pi$ is superharmonic, $\hat{\parav}_\pi$ is smaller than $\vx +\vy$,
which is the variance of the Bayesian predictive density based on the uniform prior.
It can be shown that $\Delta  m_\pi \leq 0$ holds if $\Delta \pi \leq 0$ as follows.
We have
\begin{align*}
\frac{\partial}{\partial x_i}
m_\pi(x)
&=
\int \frac{\mu-x_i}{\vx}p(x;\mu, \vx I_d)\pi(\mu) d\mu
=
\int p(x;\mu, \vx I_d) \frac{\partial}{\partial \mu_i}\pi(\mu) d\mu.
\end{align*}
The last equation comes from Stein's lemma.
Thus,
\begin{align*}
\frac{\partial^2}{\partial x_i^2}
m_\pi(x)
&=
\int p(x;\mu, \vx I_d) \frac{\partial^2}{\partial \mu_i^2}\pi(\mu) d\mu
\end{align*}
and we obtain
\begin{align*}
\Delta m_\pi(x)
&=
\int p(x;\mu, \vx I_d) \Delta\pi(\mu) d\mu.
\end{align*}
Therefore, when $\pi$ is a superharmonic function, $m_\pi$ is also superharmonic and
\begin{align*}
\hat{\parav}_\pi(x) = \vx + \vy + \hpi(x) \leq \vx + \vy.
\end{align*}
On the other hand, because 
\begin{align*}
\hat{\parav}_\pi &= \vy + (\mathrm{E}[\mu^\top \mu \mid x]-\hat{\mu}_\pi^\top\hat{\mu}_\pi)/ \mdim
= \vy + \mathrm{E}[(\mu - \hat{\mu}_\pi)^\top (\mu-\hat{\mu}_\pi) \mid x] / \mdim,
\end{align*}
we have $\hat{\parav}_\pi(x) \geq \vy$.
Because
$\mathrm{tr} \hat{\Sigma}_{\pi} = \mdim \hat{\parav}_\pi$,
the average of the eigenvalues of $\hat{\Sigma}_{\pi}$ is also smaller than the variance $u+v$
of the Bayesian predictive density based on the uniform prior.

%
%

\section{Risk for infinitesimal prediction}
We compare the Kullback--Leibler risk of the extended plug-in densities with
Bayes extended esitmators
and that of the Bayesian predictive density $p_\mathrm{U}(y \mid x)$ based on the uniform prior.
The Bayesian predictive density
is included in the normal model $\mathrm{N}_d(\mu, \parav I_d)~(\parav\in \mathbb{R})$ and it is minimax.
It is desirable to obtain predictive densities belonging to the extended models
that perform better than $p_\mathrm{U}(y \mid x)$.

The risk function of $\hat{p}(y \mid x)$ is
\[
R(\mu;\hat{p}) = {\rm E}[D\{p(y;\mu, \vy I_d);\hat{p}(y\mid x) \}]=\int p(x;\mu, \vx I_d)D\{p(y;\mu, \vy I_d);\hat{p}(y\mid x) \}dx.
\]
For the predictive densities $\hat{p}_1(y\mid x)$ and $\hat{p}_2(y\mid x)$, we have
\begin{align*}
D\{p(y;\mu, \vy I_d);\hat{p}_1(y\mid x) \} - D\{p(y;\mu, \vy I_d);\hat{p}_2(y\mid x) \}
&= \int p(y;\mu, \vy I_d) \log \frac{\hat{p}_2(y\mid x)}{\hat{p}_1(y\mid x)} dy.
\end{align*}

We introduce the time variables $s := 1/\vx$ and $\ty := 1/\vy$, 
which can be regarded as the numbers of observations and the number of future samples, respectively.
We consider a Gaussian process $Z_\tau$ $(\tau \geq 0)$ defined by the stochastic differential equation
\[
 \dd Z_\tau = \mu \dd \tau + \dd B_\tau ~~ (\tau \geq 0),
\]
where $Z_0 = 0$ and $B_\tau$ $(\tau \geq 0)$ is a standard Browninan motion.
Consequently, the distribution of $(1/\tau)Z_\tau$ is $\mathrm{N}(\mu,1/\tau)$.
Thus, our problem is equivalent to a problem in which
we observe $(1/s) Z_s$ and predict
$(1/t) (Z_{s + t} - Z_s)$.
Therefore, $s$ and $\ty$ correspond to the observation time and prediction time, respectively.
Let $\hat{\mu}_{t,\pi}$ be the posterior mean of $\mu$ based on observation $Z_t$ and prior $\pi$.

In this setting, the relationship 
between prediction risk and estimation risk used in \citet{brown2008}
is represented by
\begin{align}
R(\mu; p_{\mathrm{U}}) - R(\mu; p_{\pi})
= \int_{\tx}^{\tx +\ty} \frac{\mathrm{E}_{\tau}[\| \tau^{-1} Z_\tau - \mu\|^2]
-\mathrm{E}_{\tau}[\|\hat{\mu}_{\tau,\pi}-\mu\|^2]}{2} \dd \tau,
\label{riskintegration}
\end{align}
where
$\mathrm{E}_{\tau}[\cdot]$ means taking expectation
about $\mathrm{N}_d(\mu, 1/\tau)$.
This shows that the risk difference of the Bayesian predictive densities is represented as the integration of the estimation risk difference from $s$ to $\tx + \ty$.

The relation \eqref{riskintegration} shows that
\begin{align*}
R(\mu; p_\mathrm{U}) - R(\mu; p_\pi) \geq 0
\end{align*}
holds if
\[
\mathrm{E}_{v}[\|x-\mu\|^2]-\mathrm{E}_{v}[\|\hat{\mu}_\pi-\mu\|^2]\geq 0
\]
for all $\tau> 0$.
Thus, if $\pi$ is a superharmonic prior, $p_\pi$ dominates $p_\mathrm{U}$.
In this sense, estimation risk difference can be considered as infinitesimal-prediction risk.

Subsequently, we consider the relationship between the risk of extended plug-in densities and that of Bayes extended estimators.
We compare the risk functions of extended plug-in
predictive densities with Bayes extended estimators
based on superharmonic priors
and the uniform prior $\pi_\mathrm{U}$.
Recall that the extended plug-in densities $p(y;\hat{\mu}_\mathrm{U}, \hat{\parav}_\mathrm{U})$
and $p(y;\hat{\mu}_\mathrm{U}, \hat{\Sigma}_\mathrm{U})$
based on the uniform prior $\pi_\mathrm{U}$
coincide with the Bayesian predictive density
based on $\pi_\mathrm{U}$.
We show that the infinitesimal prediction risk difference of
extended plug-in predictive densities at $\tau = s$
is the risk difference between the corresponding Bayes extended estimators.
This shows that the extended plug-in distributions with $(\hat{\mu}_\pi,\hat{\parav}_{t,\pi})$ and $(\hat{\mu}_\pi,\hat{\Sigma}_{t,\pi})$,
where the subscript $t$ is added to the densities to clarify
their dependency on it,
have better performance than $\hat{p}_\mathrm{U}$
if $\ty$ is small enough and $\pi$ is a superharmonic prior.
From \eqref{pm1}, $\hat{\mu}_\pi$ does not depend on $t$.
\begin{thm} 
\label{thm:risk}
Denote the Kullback--Leibler risk of $p(y;\hat{\mu}_\pi, \hat{\parav}_{\ty,\pi}I_d)$
and $p(y;\hat{\mu}_\pi, \hat{\Sigma}_{\ty,\pi})$
as
\[
R_t(\mu;\hat{\mu}_\pi, \hat{\parav}_{\ty,\pi})
\coloneqq 
\int p(x;\mu, \tx^{-1} I_d) D\{p(y;\mu, \ty^{-1} I_d);
p(y;\hat{\mu}_\pi(x), \hat{\parav}_{t,\pi}(x) I_d) \}dx.
\]
and
\[
R_t(\mu;\hat{\mu}_\pi, \hat{\Sigma}_{\ty,\pi})
\coloneqq 
\int p(x;\mu, \tx^{-1} I_d) D\{p(y;\mu, \ty^{-1} I_d);
p(y;\hat{\mu}_\pi(x), \hat{\Sigma}_{\ty,\pi}(x)) \dd x,
\]
respectively.
Then,
\begin{align}
\lim_{\ty \to 0}
\frac{\partial}{\partial \ty}\left\{ R_t(\mu; \hat{p}_{t,\mathrm{U}})
- R_t(\mu; \hat{\mu}_{\pi},\hat{\parav}_{t,\pi})\right\}
&= \frac{\mathrm{E}[\|x-\mu\|^2]
-\mathrm{E}[\|\hat{\mu}_\pi-\mu\|^2]}{2}
\label{thm-e1}
\end{align}
and
\begin{align}
\lim_{\ty \to 0}
\frac{\partial}{\partial \ty} \left\{ R_t(\mu; \hat{p}_{t,\mathrm{U}})
- R_t(\mu; \hat{\mu}_{\pi},\hat{\Sigma}_{t,\pi})\right\}
&= \frac{\mathrm{E}[\|x-\mu\|^2]
-\mathrm{E}[\|\hat{\mu}_\pi-\mu\|^2]}{2}
\label{thm-e2}
\end{align}
hold.
\end{thm}
%
%
\begin{proof}

The risk difference between $p(y;\hat{\mu}_\pi, \hat{\parav}_{\ty,\pi}I_d)$
and $p_\mathrm{t,U}(y \mid x) = p(y;x,(\tx^{-1}+\ty^{-1})I_d)$ is given by
\begin{align}
\label{eq:riskdif}
R_t&(\mu; \hat{p}_{t,\mathrm{U}}) - R_t(\mu; \hat{\mu}_\pi, \hat{\parav}_{t,\pi})
= \mathrm{E}_{x,y \mid t}
\left[\log \frac{p(y;\hat{\mu}_\pi, \hat{\parav}_{t,\pi} I_d)}
{p(y;x,(\tx^{-1}+\ty^{-1})I_d)} \right] \nonumber \\
&= \mathrm{E}_{x,y \mid t}\left[
-\frac{\mdim}{2}\log\frac{\hat{\parav}_{t,\pi}}{\tx^{-1}+\ty^{-1}}
-\frac{1}{2 \hat{\parav}_{t,\pi}}(y-\hat{\mu}_\pi)^\top (y-\hat{\mu}_\pi)
+ \frac{1}{2(\tx^{-1}+\ty^{-1})}(y-x)^\top(y-x)
\right] \notag \\
&= \mathrm{E}_{x,y \mid t}\left[
-\frac{\mdim}{2}\log\frac{\hat{\parav}_{t,\pi}}{\tx^{-1}+\ty^{-1}}
-\frac{1}{2 \hat{\parav}_{t,\pi}}(y-\hat{\mu}_\pi)^\top (y-\hat{\mu}_\pi) \right]
+ \frac{d}{2},
\end{align}
where the expectation about $(x,y)$ is denoted as $\mathrm{E}_{x,y \mid t}[\,\cdot\,]$.
We evaluate the differential of the risk difference with respect to $\ty$.
From \eqref{hpi}, we have
\begin{align*}
\frac{\partial \hat{\parav}_{t,\pi}}{\partial \ty}
= \frac{\partial}{\partial \ty}\{\tx^{-1} + \ty^{-1} + \hpi(x)\}
= -\ty^{-2}.
\end{align*}
Thus, 
\begin{align}
\label{logvarratio}
\frac{\partial}{\partial \ty}\log\frac{\hat{\parav}_{t,\pi}}{\tx^{-1} + \ty^{-1}}
&= \frac{1}{\hat{\parav}_{t,\pi}}\frac{\partial \hat{\parav}_{t,\pi}}{\partial \ty}
- \frac{1}{\tx^{-1} + \ty^{-1}}\frac{\partial (\tx^{-1} + \ty^{-1})}{\partial \ty}
= -\ty^{-2}\left(
\frac{1}{\hat{\parav}_{t,\pi}} - \frac{1}{\tx^{-1} + \ty^{-1}}
\right).
\end{align}
We differentiate the rest of (\ref{eq:riskdif}) and  obtain
\begin{align}
\label{varratio}
\frac{\partial}{\partial \ty}&\mathrm{E}_{x,y \mid t}\left[
- \frac{1}{2 \hat{\parav}_{t,\pi}}(y-\hat{\mu}_\pi)^\top (y-\hat{\mu}_\pi)
\right]
= -\frac{1}{2}\frac{\partial}{\partial \ty}\mathrm{E}_{x}\left[
\frac{d \ty^{-1} + (\mu-\hat{\mu}_\pi)^\top (\mu-\hat{\mu}_\pi)}{\hat{\parav}_{t,\pi}}
\right] \notag \\
&= -\frac{1}{2}\mathrm{E}_{x}\left[
- d \ty^{-2} \frac{1}{\hat{\parav}_{t,\pi}}
+ d \ty^{-1} \frac{\ty^{-2}}{\hat{\parav}_{t,\pi}^{\,2}}
+ \frac{\ty^{-2}}{\hat{\parav}_{t,\pi}^{\,2}}(\mu-\hat{\mu}_\pi)^\top (\mu-\hat{\mu}_\pi)
\right].
\end{align}
From \eqref{eq:riskdif}, \eqref{logvarratio}, and \eqref{varratio}, we obtain
\begin{align*}
&\frac{\partial}{\partial \ty} \left\{ R_t(\mu; \hat{p}_{t,{\mathrm{U}}})
- R_t(\mu; \hat{\mu}_\pi, \hat{\parav}_{t,\pi}) \right\} \\
&= \frac{1}{2}\mathrm{E}_{x} \left[ \mdim \ty^{-2}\left(\frac{1}{\hat{\parav}_{t,\pi}} - \frac{1}{\tx^{-1} + \ty^{-1}} \right)
+ \mdim \ty^{-2} \frac{1}{\hat{\parav}_{t,\pi}}
- \mdim \ty^{-1} \frac{\ty^{-2}}{{\hat{\parav}_{t,\pi}}^{\,2}}
- \frac{\ty^{-2}}{{\hat{\parav}_{t,\pi}}^{\,2}}(\mu-\hat{\mu}_\pi)^\top (\mu-\hat{\mu}_\pi)
\right]\\
&=\frac{1}{2}\mathrm{E}_{x} \Biggl[
\mdim \ty^{-2} \left(
\frac{2}{\tx^{-1}+\ty^{-1}+h_\pi} - \frac{1}{\tx^{-1}+\ty^{-1}}
\right)
- \mdim \ty^{-1} \frac{\ty^{-2}}{(\tx^{-1}+\ty^{-1}+h_\pi)^2}
- \frac{\ty^{-2}}{(\tx^{-1}+\ty^{-1}+h_\pi)^2}(\mu-\hat{\mu}_\pi)^\top (\mu-\hat{\mu}_\pi) \Biggr] \\
&=
\frac{1}{2}\mathrm{E}_{x} \Biggl[
{\mdim}\ty^{-1} \left\{
2-\frac{2(\tx^{-1}+h_\pi)}{\tx^{-1}+\ty^{-1}+h_\pi} - 1 + \frac{\tx^{-1}}{\tx^{-1}+\ty^{-1}}
\right\}
- \mdim \ty^{-1}\left\{1-\frac{\ty^2(\tx^{-1}+h_\pi)^2
+ 2\ty(\tx^{-1}+h_\pi)}{(1 + \ty \tx^{-1} + \ty h_\pi)^2} \right\}\\
&~~~~~~~~~~~~ - 
\frac{\ty^{-2}}{(\tx^{-1}+\ty^{-1}+h_\pi)^2}(\mu-\hat{\mu}_\pi)^\top (\mu-\hat{\mu}_\pi)
\Biggr] \\
&=
\frac{1}{2}\mathrm{E}_{x} \Biggl[
\mdim \left\{
-\frac{2(\tx^{-1}+h_\pi)}
{1 + \ty \tx^{-1} + \ty h_\pi} + \frac{\tx^{-1}}{1 + \ty \tx^{-1}}
+ \frac{\ty(\tx^{-1}+h_\pi)^2
+ 2 (\tx^{-1}+h_\pi)}{(1 + \ty \tx^{-1} + \ty h_\pi)^2} \right\}
- \frac{1}{(1 + \ty \tx^{-1} + \ty h_\pi)^2}(\mu-\hat{\mu}_\pi)^\top (\mu-\hat{\mu}_\pi)
\Biggr] \\
&=
\frac{1}{2}\mathrm{E}_{x} \Biggl[
\mdim \left\{ \tx^{-1} - \frac{\ty \tx^{-2}}{1 + \ty \tx^{-1}}
- \frac{\ty(\tx^{-1}+h_\pi)^2}{(1 + \ty \tx^{-1} + \ty h_\pi)^2}
\right\}
- \frac{1}{(1 + \ty \tx^{-1} + \ty h_\pi)^2}(\mu-\hat{\mu}_\pi)^\top (\mu-\hat{\mu}_\pi)
\Biggr].
\end{align*}
Thus, from
\begin{align*}
&\lim_{\ty \to 0} \frac{\partial}{\partial \ty} \left\{ R_t(\mu; \hat{p}_{t,{\mathrm{U}}})
- R_t(\mu; \hat{\mu}_\pi, \hat{\parav}_{t,\pi}) \right\}
= \frac{\mdim \tx^{-1} - \mathrm{E}_{x}[\|\hat{\mu}_\pi-\mu\|^2]}{2},
\end{align*}
the desired result \eqref{thm-e1} is obtained.

Next, the risk difference between the extended plug-in density $p(y;\hat{\mu}_\pi, \hat{\Sigma}_{t,\pi})$
and $p_\mathrm{U}(y \mid x)$ is
\begin{align}
&R_t(\mu; p_\mathrm{U}) - R_t(\mu; \hat{\mu}_\pi, \hat{\Sigma}_{t,\pi})
= \mathrm{E}_{x,y \mid t}
\left[\log \frac{p(y;\hat{\mu}_\pi, \hat{\Sigma}_{t,\pi})}
{p(y;x,(\tx^{-1}+\ty^{-1})I_d)}
\right] \nonumber \\
&= \mathrm{E}_{x,y \mid t}\left[
-\frac{1}{2}\log\frac{|\hat{\Sigma}_{t,\pi}|}{(\tx^{-1} +\ty^{-1})^d}
- \frac{1}{2}(y-\hat{\mu}_\pi)^\top\hat{\Sigma}_{t,\pi}^{-1} (y-\hat{\mu}_\pi)
+ \frac{1}{2(\tx^{-1} +\ty^{-1})}(y-x)^\top(y-x)
\right] \notag \\
&= \mathrm{E}_{x,y \mid t}\left[
- \frac{1}{2}\log\frac{|\hat{\Sigma}_{t,\pi}|}{(\tx^{-1} +\ty^{-1})^d}
- \frac{1}{2}(y-\hat{\mu}_\pi)^\top\hat{\Sigma}_{t,\pi}^{-1} (y-\hat{\mu}_\pi)
\right]
+ \frac{\mdim}{2}.
\end{align}
From \eqref{Sigmapi},
\begin{align}
\frac{\partial}{\partial \ty} \hat{\Sigma}_{t,\pi}
&= \frac{\partial}{\partial t} \bigl\{(\tx^{-1} + \ty^{-1}) I_d + H_\pi(x) \bigr\}
= - \ty^{-2} I_d.
\end{align}
Thus,
\begin{align*}
\frac{\partial}{\partial \ty} \log|\hat{\Sigma}_{t,\pi}|
= \mathrm{tr} \Bigl\{ \hat{\Sigma}_{t,\pi}^{-1}(-\ty^{-2})I_d \Bigr\}
= -\ty^{-2}\mathrm{tr} \hat{\Sigma}_{t,\pi}^{-1}
\end{align*}
and
\begin{align*}
\frac{\partial}{\partial \ty} & \mathrm{E}_{x,y \mid t}
[ (y-\hat{\mu}_\pi)^\top\hat{\Sigma}_{t,\pi}^{-1} (y-\hat{\mu}_\pi)]
=\frac{\partial}{\partial \ty} \mathrm{E}_{x}[
\ty^{-1}\mathrm{tr}(\hat{\Sigma}_{t,\pi}^{-1}) + (\mu-\hat{\mu}_\pi)^\top\hat{\Sigma}_{t,\pi}^{-1} (\mu-\hat{\mu}_\pi)
]\\
&= -\ty^{-2}\mathrm{tr} \, \hat{\Sigma}_{t,\pi}^{-1}
+ \ty^{-1}\mathrm{tr}(\ty^{-2}\hat{\Sigma}_{t,\pi}^{-2}) + \ty^{-2}(\mu-\hat{\mu}_\pi)^\top\hat{\Sigma}_{t,\pi}^{-2} (\mu-\hat{\mu}_\pi).
\end{align*}
Therefore, we obtain
\begin{align*}
&\frac{\partial}{\partial \ty}\mathrm{E}_{x,y \mid t}
\left[\log \frac{p(y;\hat{\mu}_\pi, \hat{\Sigma}_{t,\pi})}
{p(y;x,(\tx^{-1}+\ty^{-1})I_d)}
\right] \nonumber \\
&=\frac{\partial}{\partial \ty} \mathrm{E}_{x,y \mid t}
\left[ -\frac{1}{2}\log\frac{|\hat{\Sigma}_{t,\pi}|}{(\tx^{-1}+\ty^{-1})^d}
-\frac{1}{2}(y-\hat{\mu}_\pi)^\top\hat{\Sigma}_{t,\pi}^{-1} (y-\hat{\mu}_\pi)
\right]\\
&=  \mathrm{E}_{x} \Biggl[
-\frac{1}{2} \Bigl( -\ty^{-2}\mathrm{tr} \, \hat{\Sigma}_{t,\pi}^{-1}
- \mdim \frac{-\ty^{-2}}{\tx^{-1}+\ty^{-1}} \Bigr)
- \frac{1}{2}
\left\{ -\ty^{-2}\mathrm{tr} \, \hat{\Sigma}_{t,\pi}^{-1}
+ \ty^{-1} \mathrm{tr} (\ty^{-2}\hat{\Sigma}_{t,\pi}^{-2}) + \ty^{-2}(\mu-\hat{\mu}_\pi)^\top\hat{\Sigma}_{t,\pi}^{-2} (\mu-\hat{\mu}_\pi) \right\}
\Biggr] \\
&=  \mathrm{E}_{x} \Biggl[
\frac{\mdim \ty^{-1}}{2} \Bigl( \frac{\tx^{-1}}{\tx^{-1} +\ty^{-1}}-1 \Bigr)
+ \ty^{-1} \mathrm{tr}(\ty^{-1}\hat{\Sigma}_{t,\pi}^{-1})
- \frac{\ty^{-1}}{2} \mathrm{tr}(\ty^{-2}\hat{\Sigma}_{t,\pi}^{-2})
- \frac{1}{2}\ty^{-2}(\mu-\hat{\mu}_\pi)^\top\hat{\Sigma}_{t,\pi}^{-2} (\mu-\hat{\mu}_\pi)
\Biggr].
\end{align*}
Let
\[
A_\pi := {\tx}^{-1} I_d + {\tx^{-2}}\left(\frac{\Delta^2 m_\pi}{m_\pi} - \frac{\Delta m_\pi\Delta m_\pi^\top}{m_\pi^2} \right).
\]
Then,
\[
\ty \hat{\Sigma}_{t,\pi} = I_d +\ty A_\pi
\]
and
\[
\ty^{-1} \hat{\Sigma}_{t,\pi}^{-1} = (I_d +\ty A_\pi)^{-1}.
\]
When $\ty$ is small enough, all absolute values of the eigenvalues of $\ty A_\pi$ are smaller than 1 and
\[
\ty^{-1} \hat{\Sigma}_{t,\pi}^{-1} = \sum_{i=0}^\infty (-1)^i(\ty A_\pi)^{i}.
\]
In the same manner, let
\[
B_\pi := 2A_\pi +\ty A_\pi^2,
\]
and we have
\[
(\ty \hat{\Sigma}_{t,\pi})^2 = I_d + \ty B_\pi,
\]
and when $\ty$ is small enough that all absolute values of the eigenvalues of $\ty B_\pi$ are smaller than 1,
\[
\ty^{-2} \hat{\Sigma}_{t,\pi}^{-2} = (I_d + \ty B_\pi)^{-1} = \sum_{i=0}^\infty (-1)^i(\ty B_\pi)^{i}.
\]
Therefore, when $\ty>0$ is small enough,
\begin{align*}
&\frac{\partial}{\partial \ty}\mathrm{E}_{x,y \mid t}
\biggl[\log \frac{p(y;\hat{\mu}_\pi, \hat{\Sigma}_{t,\pi})}
{p(y;x,(\tx^{-1}+\ty^{-1})I_d)} \biggr] \notag \\
&= \mathrm{E}_{x} \Biggl[
\frac{d \ty^{-1}}{2} \Bigl(\frac{\tx^{-1}}{\tx^{-1} +\ty^{-1}}-1 \Bigr)
+ \ty^{-1} \mathrm{tr}(\ty^{-1}\hat{\Sigma}_{t,\pi}^{-1})
- \frac{\ty^{-1}}{2} \mathrm{tr}(\ty^{-2}\hat{\Sigma}_{t,\pi}^{-2})
- \frac{1}{2}\ty^{-2}(\mu-\hat{\mu}_\pi)^\top\hat{\Sigma}_{t,\pi}^{-2} (\mu-\hat{\mu}_\pi)
\Biggr] \\
&= \frac{\mdim \ty^{-1}}{2} \frac{\tx^{-1}}{\tx^{-1}+\ty^{-1}}
+ \mathrm{E}_{x} \Biggl[
-\frac{d \ty^{-1}}{2}
+ \ty^{-1} \mathrm{tr} \Bigl\{ \sum_{i=0}^\infty(-1)^i (\ty A_\pi)^i \Bigr\}
- \frac{\ty^{-1}}{2} \mathrm{tr} \Bigl\{ \sum_{j=0}^\infty (-1)^j (\ty B_\pi)^j \Bigr\} \\
&~~~~
- \frac{1}{2}(\mu-\hat{\mu}_\pi)^\top \Bigl\{ \sum_{j=0}^\infty (-1)^j (\ty B_\pi)^j \Bigr\} (\mu-\hat{\mu}_\pi)
\Biggr] \\
&= \frac{\mdim}{2} \frac{\tx^{-1}}{1 + \ty \tx^{-1}}
+ \mathrm{E}_{x} \Biggl[
-\frac{\mdim \ty^{-1}}{2}
+ \mathrm{tr} \Bigl\{\ty^{-1} I_d - A_\pi + \sum_{i=2}^\infty(-1)^i \ty^{i-1}A_\pi^i \Bigr\}
- \frac{1}{2}\mathrm{tr}
\Bigl\{\ty^{-1} I_d - 2A_\pi -\ty A_\pi^2 + \sum_{j=2}^\infty (-1)^j \ty^{j-1} B_\pi^j \Bigr\} \\
&~~~~ - \frac{1}{2}(\mu-\hat{\mu}_\pi)^\top
\Bigl\{ I_d + \sum_{j=1}^\infty (-1)^j (\ty B_\pi)^j \Bigr\} (\mu-\hat{\mu}_\pi)
\Biggr].
\end{align*}
Thus, from
\begin{align*}
&\lim_{t \to 0} \frac{\partial}{\partial \ty}\mathrm{E}_{x,y \mid t}
\biggl[\log \frac{p(y;\hat{\mu}_\pi, \hat{\Sigma}_{t,\pi})}
{p(y;x,(\tx^{-1}+\ty^{-1})I_d)} \biggr]
= \frac{\mdim \tx^{-1} - \mathrm{E}_{x}[\|\hat{\mu}_\pi-\mu\|^2]}{2},
\end{align*}
the desired result \eqref{thm-e2} is obtained.
\end{proof}
%

%
%
\section{Numerical experiments}
We compare the Kullback--Leibler risks of the extended plug-in densities based on Stein's prior $\pi_\mathrm{S}$,
the Bayesian predictive density $p_\mathrm{U}(y \mid x)$ based on the uniform prior $\pi_\mathrm{U}$,
the Bayesian predictive density $p_\mathrm{S}(y \mid x)$ based on $\pi_\mathrm{S}$,
and an empirical Bayes method studied in \citet{xu2011}.
Observation $x$ is distributed according to
$\mathrm{N}_d(\mu, \vx I_d)$ with $\mu\in\mathbb{R}^d$ and $\vx > 0$,
and a future sample $y$ comes from a normal distribution $\mathrm{N}_d(\mu, \vy I_d)$
with the same mean $\mu$ and with a possibly different variance $\vy > 0$.
In Theorem \ref{thm:risk}, we observe that the proposed methods based on a superharmonic prior dominate
$p_\mathrm{U}(y \mid x)$ when $1/\vy$ is close to 0.
In this experiments, we numerically evaluate the Kullback--Leibler risks for finite $\vy > 0$. 
Although we are interested in the risk comparison among predictive densities that can be obtained by simple computations,
we also simulate the Kullback--Leibler risk of the Bayesian predictive density $p_\mathrm{S}(y \mid x)$
based on $\pi_\mathrm{S}$
to verify the approximate performance of those plug-in densities.

When Stein's prior is employed, the extended estimators $\hat{\mu}_\pi$, $\hat{\parav}_\pi$ and $\hat{\Sigma}_{\pi}$ 
are given by
\begin{align*}
\hat{\mu}_\pi &= F_1 x,\\
\hat{\xi}_\pi &= v + F_1 u + \frac{x^\top x}{d}(F_2 - F_1^2),\\
\hat{\Sigma}_\pi &= vI_d + F_1 uI_d + {x^\top x}(F_2 - F_1^2)
\end{align*}
where
\begin{align*}
    F_1 &= 1 - 2\frac{\phi_{d+2}(\|x\|/\sqrt{u})}{\phi_{d}(\|x\|/\sqrt{u})},\\
    F_2 &= 1 + 4\frac{\phi_{d+4}(\|x\|/\sqrt{u})-\phi_{d+2}(\|x\|/\sqrt{u})}{\phi_{d}(\|x\|/\sqrt{u})}
\end{align*}%
and
\[
\phi_d(a) = a^{-d+2}\int^{a^2/2}_{0} s^{d/2-2}\exp(-s) ds~~~(a\geq 0).
\]
These evaluations of the extended estimators follow the mixture representation \eqref{integral representation}
of Stein's prior.
For comparison, we employed the empirical Bayes method $\hat{p}_{p-3}$
from the numerical analysis in \citet{xu2011}.
The Kullback--Leibler risks are computed by taking the average of 5000 trials.

The simulation results are shown in Figure \ref{fig_KLrisk}.
As expected, the risk of $p_\mathrm{S}(y \mid x)$ is the smallest, whereas the risk of $p_\mathrm{U}(y \mid x)$,
which is the only method in this experiment that does not employ a shrinkage prior, is the largest.
The risk of $p_\mathrm{U}(y \mid x)$ is much larger than that of any other methods in Figure \ref{p100}.
The four competitors that approximate $p_\mathrm{S}(y \mid x)$ are
the two extended plug-in densities, the empirical Bayes predictive density, and $p_\mathrm{U}(y \mid x)$.
Among these,
the extended plug-in density $p(y;\hat{\mu}_\pi, \hat{\Sigma}_{\pi})$ exhibits the best performance 
unless $\|\mu\|$ is very close to $0$.
The risk performance of the proposed extended plug-in densities approaches that of $p_\mathrm{S}(y \mid x)$
more rapidly than the empirical Bayes as $\|\mu\|$ increases.

Figure \ref{fig:three} presents the effect of the choice of the  extended models
by showing the risk differences of $p(y;\hat{\mu}_\pi, \hat{\parav}_\pi)$, $p(y;\hat{\mu}_\pi$,
$\hat{\Sigma}_{\pi})$, and $\hat{p}_\mathrm{S}(y \mid x)$.
The extended spaces to which extended plug-in densities $p(y;\hat{\mu}_\pi, \hat{\parav}_\pi)$ and $p(y;\hat{\mu}_\pi, \hat{\Sigma}_{t,\pi})$
belong are $\mathrm{N}_d(\mu, \parav I_d)$ and $\mathrm{N}_d(\mu, \Sigma)$, respectively,
and their dimensions are $\mdim+1$ and $\mdim + \mdim(\mdim+1)/2 = \mdim^2/2 + (3/2)\mdim$, respectively.
The Bayesian predictive density $p_\mathrm{S}(y \mid x)$ does not belong to any of the finite-dimensional models.
The risk comparison demonstrates that $p(y;\hat{\mu}_\pi, \hat{\Sigma}_{\pi})$ performs slightly better than
$p(y;\hat{\mu}_\pi, \hat{\parav}_\pi I_d)$, which suggests that a larger extended model result in a better performance.

\begin{figure}[H]
\begin{subfigure}{0.48\columnwidth}
\centering
\includegraphics[width=\columnwidth]{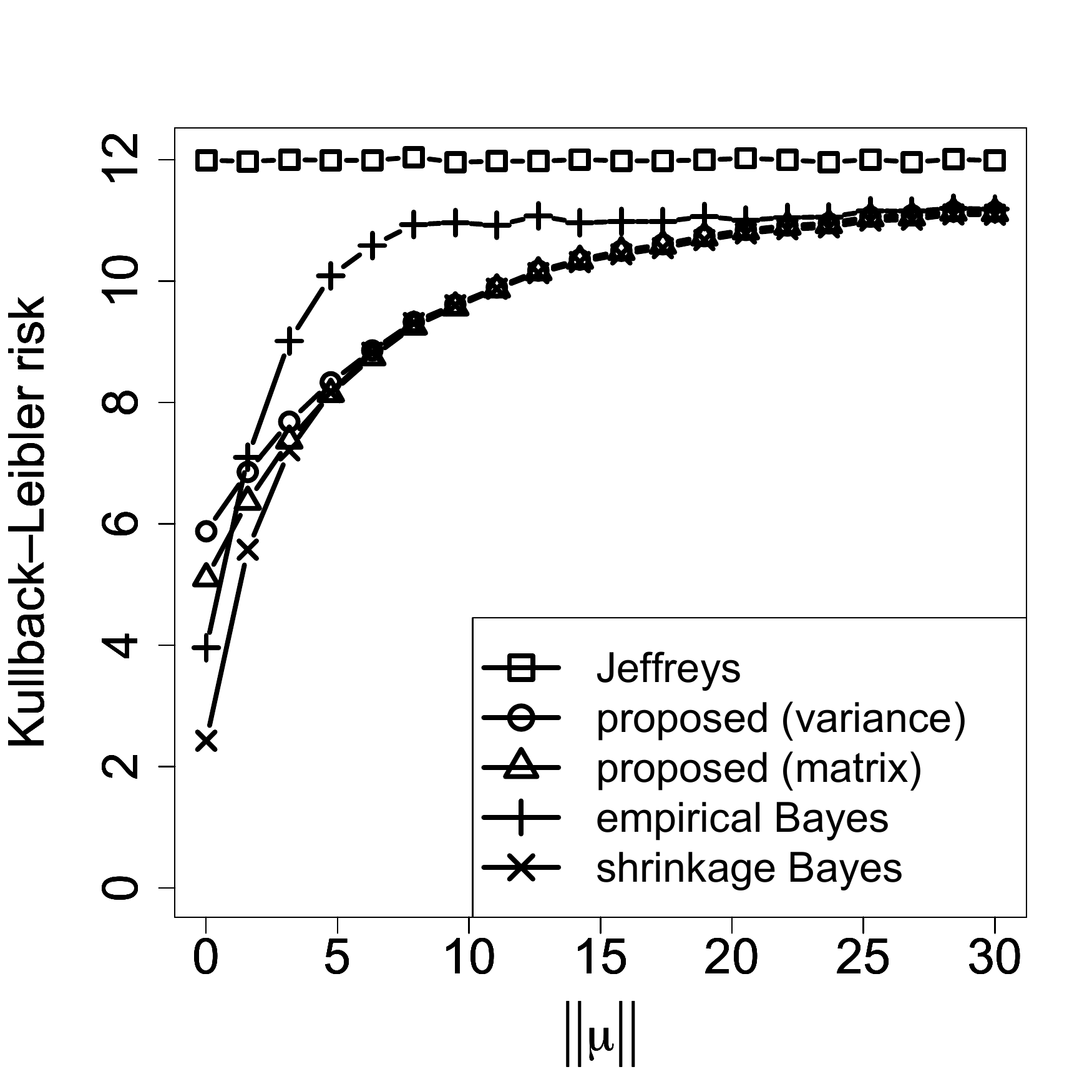}
\caption{$\mdim=10, \vx=1, \vy=1/10$}
\label{p10}
\end{subfigure}
 \ 
\centering
\begin{subfigure}{0.48\columnwidth}
\includegraphics[width=\columnwidth]{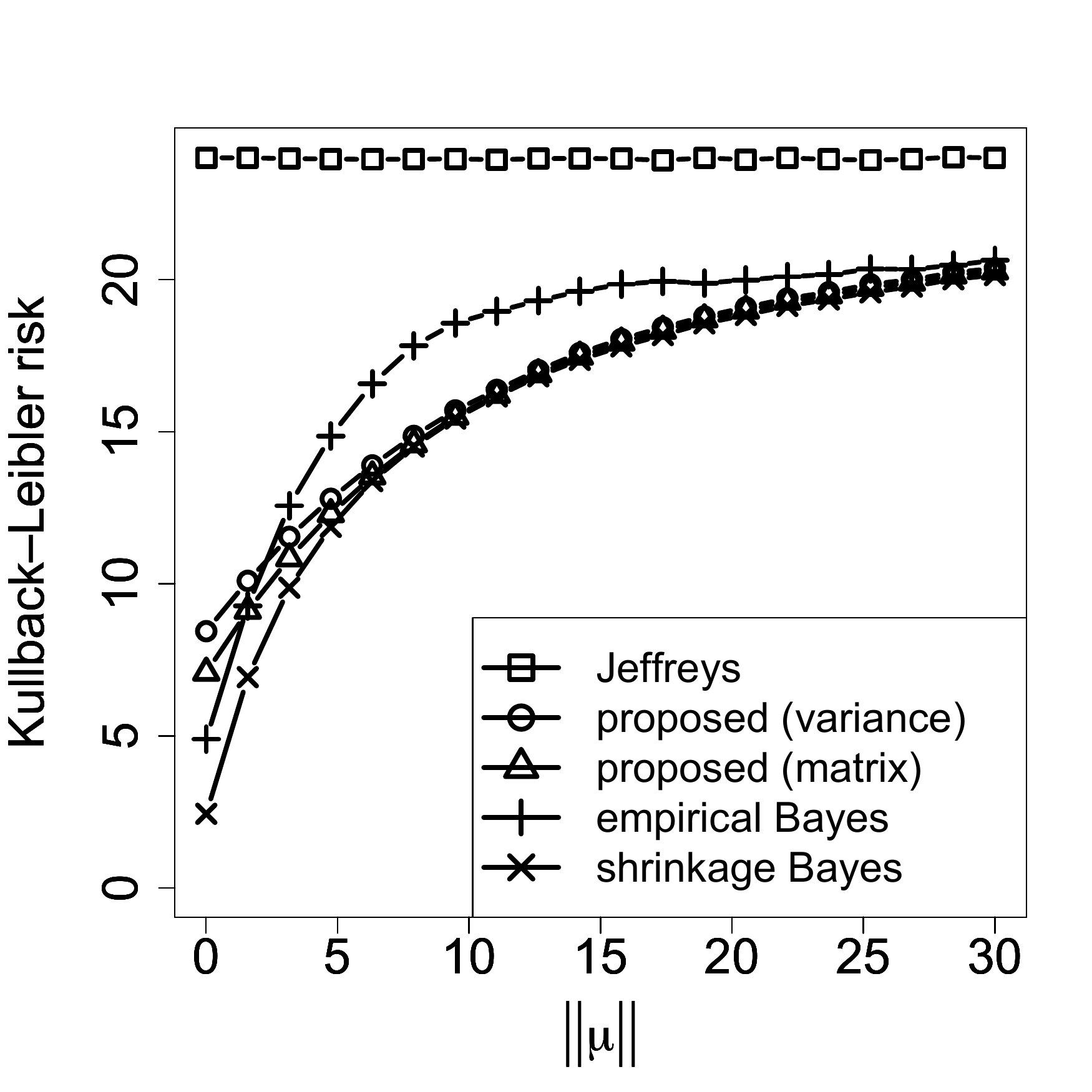}
\caption{$\mdim=20, \vx=1, \vy=1/10$}
\label{p20}
\end{subfigure}
\\
\centering
\begin{subfigure}{0.48\columnwidth}
\includegraphics[width=\columnwidth]{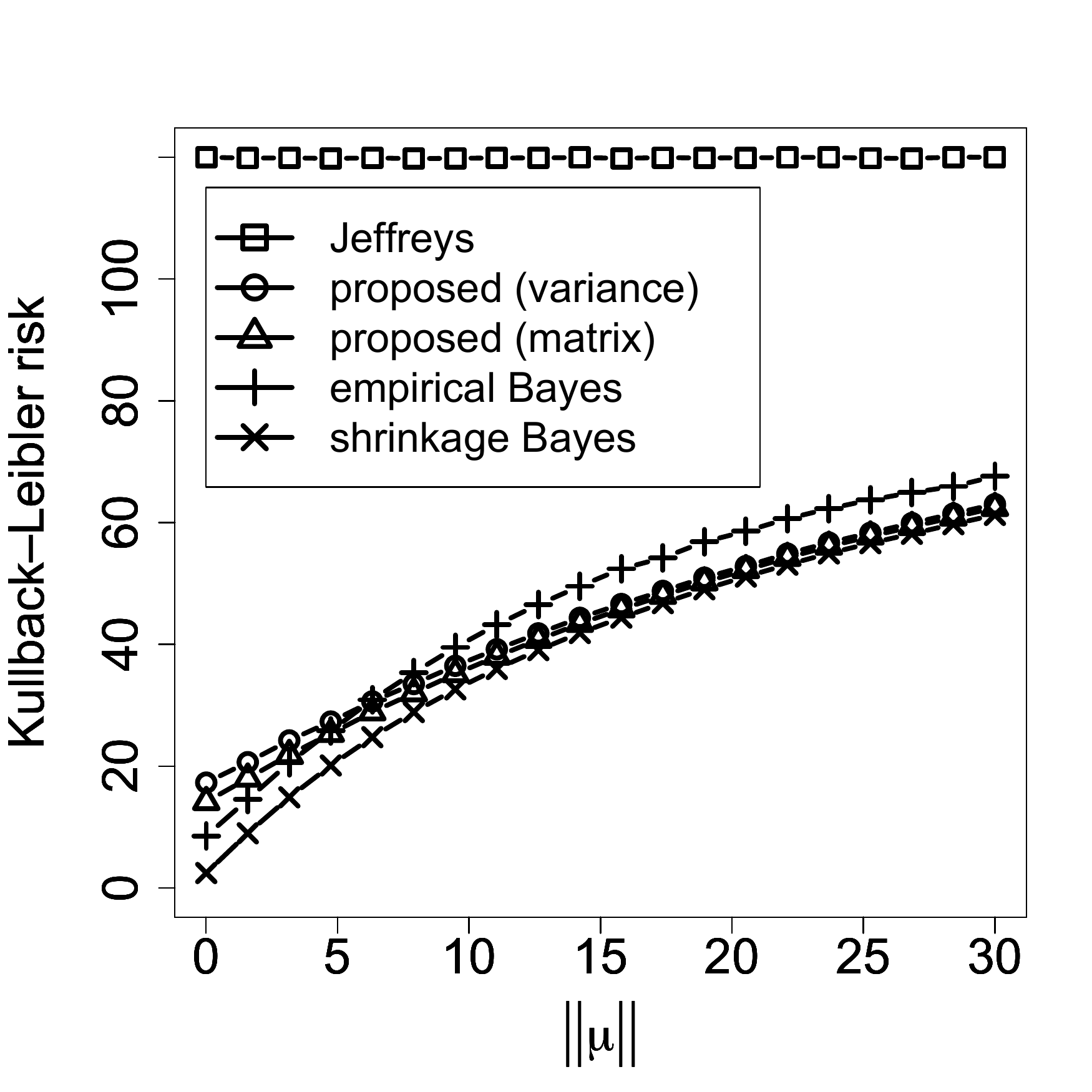}
\caption{$\mdim=100, \vx=1, \vy=1/10$}
\label{p100}
\end{subfigure}
 \ 
\centering
\begin{subfigure}{0.48\columnwidth}
\includegraphics[width=\columnwidth]{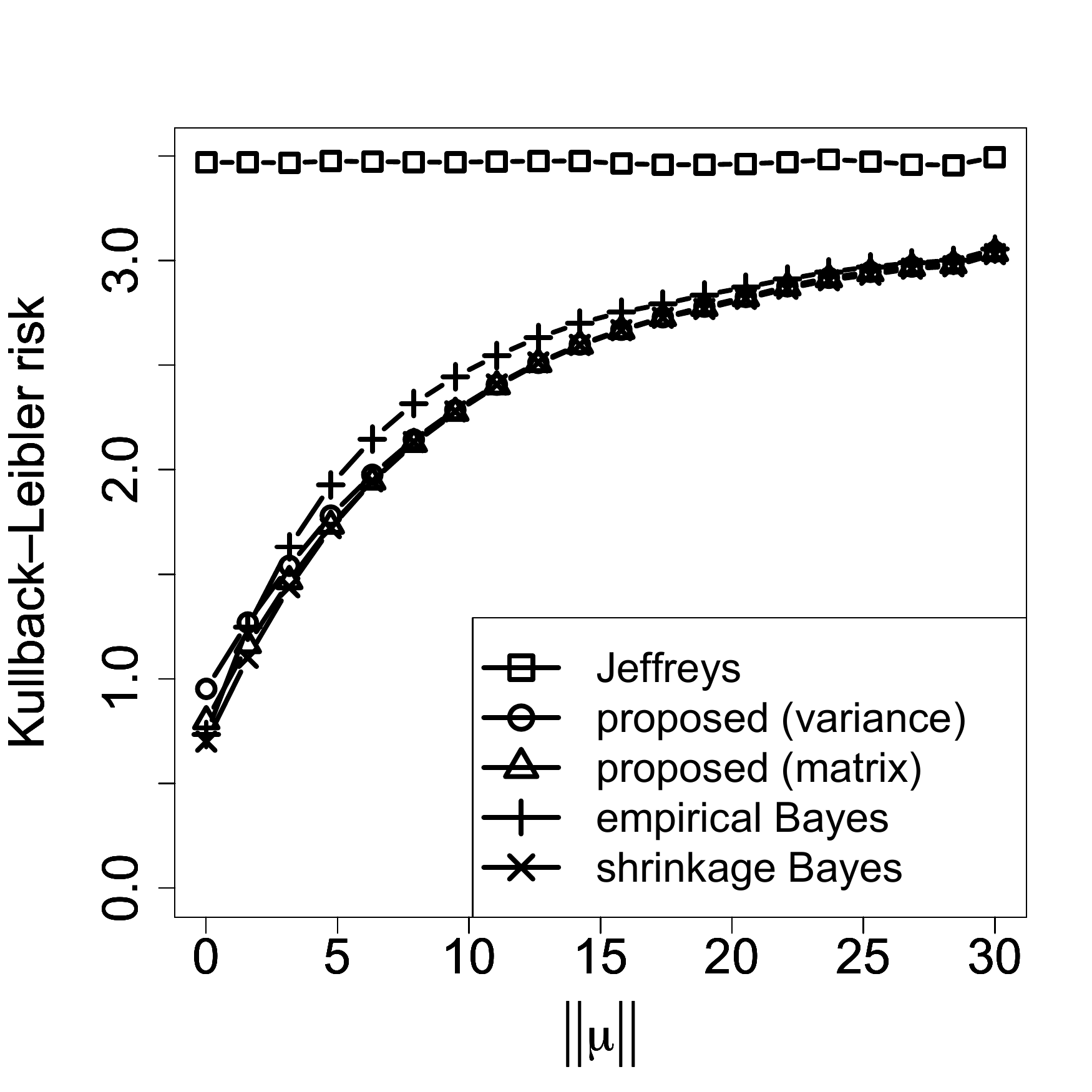}
\caption{$\mdim=10, \vx=1, \vy=1$}
\label{p10_ny1}
\end{subfigure}
\caption{
Kullback--Leibler risks of extended plugin densities with $(\hat{\mu}_\pi, \hat{\xi}_\pi)$ and $(\hat{\mu}_\pi, \hat{\Sigma}_\pi)$, empirical Bayes method in \cite{xu2011}, and Bayesian predictive densities $p_\mathrm{U}$ and $p_\mathrm{S}$
}
\label{fig_KLrisk}
\end{figure}
\begin{figure}[h]
\centering
\includegraphics[width=0.5\columnwidth]{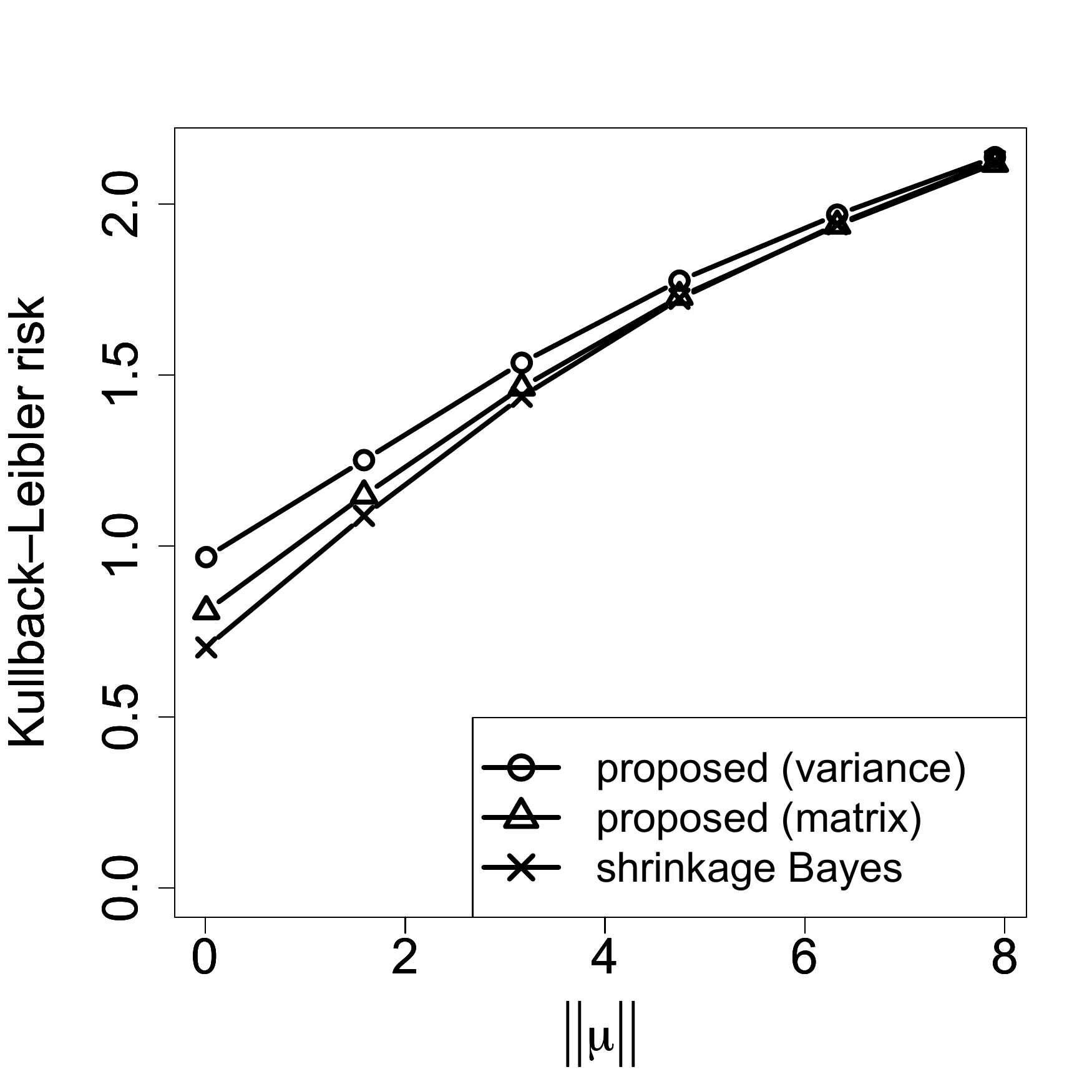}
\caption{Kullback--Leibler risks of extended plugin densities with $(\hat{\mu}_\pi, \hat{\xi}_\pi)$, $(\hat{\mu}_\pi, \hat{\Sigma}_\pi)$ and a Bayesian predictive density $p_\mathrm{S}$ when $\mdim=10, \vx=1, \vy=1$}
\label{fig:three}
\end{figure}
%

\section*{Acknowledgments}
This work was supported in part by JSPS KAKENHI Grant Numbers JP20K23316
and JP22H00510.

\bibliographystyle{biometrika}
\bibliography{eplugin_normal}

\newcommand{\noop}[1]{}
\begin{thebibliography}{11}
\expandafter\ifx\csname natexlab\endcsname\relax\def\natexlab#1{#1}\fi

\bibitem[{Brown(1971)}]{brown1971}
\textsc{Brown, L.~D.} (1971).
\newblock Admissible estimators, recurrent diffusions, and insoluble boundary
  value problems.
\newblock \textit{The Annals of Mathematical Statistics} \textbf{42}, 855--903.

\bibitem[{Brown et~al.(2008)Brown, George \& X.}]{brown2008}
\textsc{Brown, L.~D.}, \textsc{George, E.~I.} \& \textsc{X., X.} (2008).
\newblock Admissible predictive density estimation.
\newblock \textit{The Annals of Statistics} \textbf{36}, 1156--1170.

\bibitem[{George et~al.(2006)George, Liang \& Xu}]{george2006}
\textsc{George, E.~I.}, \textsc{Liang, F.} \& \textsc{Xu, X.} (2006).
\newblock Improved minimax predictive densities under
  {K}ullback^^e2^^80^^93{L}eibler loss.
\newblock \textit{{Annals of Statistics}} \textbf{34}, 78--91.

\bibitem[{George \& Xu(2008)}]{george2008}
\textsc{George, E.~I.} \& \textsc{Xu, X.} (2008).
\newblock Predictive density estimation for multiple regression.
\newblock \textit{Econometric Theory} \textbf{24}, 528--544.

\bibitem[{John(1978)}]{john1978}
\textsc{John, F.} (1978).
\newblock \textit{Partial Differential Equations}.
\newblock New York: Springer, 3rd ed.

\bibitem[{Kobayashi \& Komaki(2008)}]{kobayashi2008}
\textsc{Kobayashi, K.} \& \textsc{Komaki, F.} (2008).
\newblock Bayesian shrinkage prediction for the regression problem.
\newblock \textit{Journal of multivariate analysis} \textbf{99}, 1888--1905.

\bibitem[{Komaki(2001)}]{komaki2001}
\textsc{Komaki, F.} (2001).
\newblock A shrinkage predictive distribution for multivariate normal
  observables.
\newblock \textit{Biometrika} \textbf{88}, 859--864.

\bibitem[{Matsuda \& Komaki(2015)}]{matsuda2015}
\textsc{Matsuda, T.} \& \textsc{Komaki, F.} (2015).
\newblock Singular value shrinkage priors for bayesian prediction.
\newblock \textit{Biometrika} \textbf{102}, 843--854.

\bibitem[{Okudo \& Komaki(2021)}]{okudo2021}
\textsc{Okudo, M.} \& \textsc{Komaki, F.} (2021).
\newblock Bayes extended estimators for curved exponential families.
\newblock \textit{IEEE Transactions on Information Theory} \textbf{67},
  1088--1098.

\bibitem[{Stein(1974)}]{stein1974}
\textsc{Stein, C.} (1974).
\newblock Estimation of the mean of a multivariate normal distribution.
\newblock \textit{Proceedings of the Prague Symposium on Asymptotic Statistics}
  , 345--381.

\bibitem[{Xu \& Zhou(2011)}]{xu2011}
\textsc{Xu, X.} \& \textsc{Zhou, D.} (2011).
\newblock Empirical {B}ayes predictive densities for high-dimensional normal
  models.
\newblock \textit{{Journal of Multivariate Analysis}} \textbf{102}, 1417--1428.

\end{thebibliography}

\end{document}